\documentclass[a4paper,UKenglish,cleveref, linesno,autoref,thm-restate]{lipics-v2021}
\usepackage[utf8]{inputenc}
\usepackage{xspace}
\usepackage{graphicx}
\usepackage{amsthm,amssymb}
\usepackage[dvipsnames]{xcolor}
\usepackage{enumitem}
\usepackage{url}
\usepackage{hyperref}

\newtheorem{problem}[theorem]{Problem}
\newtheorem{oproblem}[theorem]{Open Problem}

\newcommand{\etal}{{\it{et al.}}\xspace}
\newcommand{\R}{\ensuremath{\mathbb R}}
\renewcommand{\S}{\ensuremath{\mathbb S}}
\newcommand{\nice}{{smooth}\xspace}

\newcommand{\RR}{\ensuremath{\mathcal R}\xspace}
\def\from{\mathrel\subset\mkern-10mu\joinrel\sim}

\newcommand{\CH}{\ensuremath{\mathrm {CH}}}

\newcommand{\myremark}[4]{\textcolor{blue}{\textsc{#1 #2:}} \textcolor{#4}{\textsf{#3}}}
\renewcommand{\myremark}[4]{}

\newcommand{\maarten}[2][says]{\myremark{Maarten}{#1}{#2}{WildStrawberry}}

\newcommand{\remove}[1]{}

\title{Preprocessing Disks for Convex Hulls, Revisited}

\author{Maarten L\"offler}
    {Department of Information and Computing Sciences; Utrecht University, the Netherlands}
    {m.loffler@uu.nl}%
    {https://orcid.org/0009-0001-9403-8856}%
    {}

   \author{Benjamin Raichel}{Department of Computer Science;
      University of Texas at Dallas, USA \and
      \url{http://utdallas.edu/\string~benjamin.raichel} }
   {benjamin.raichel@utdallas.edu}%
   {{https://orcid.org/0000-0001-6584-4843}}%
   {Work on this paper was partially supported by NSF CAREER Award
      1750780 and NSF Award 2311179.}

\authorrunning{M. L\"offler and B. Raichel} 

\Copyright{Maarten L\"offler and Ben Raichel} 

\ccsdesc[500]{Theory of computation~Computational geometry}%

\keywords{Convex hull, uncertainty, preprocessing} 

\category{} 

\relatedversion{} 

\acknowledgements{
The authors would like to thank Ivor van der Hoog, Vahideh Keikha, and Frank Staals for inspiring discussions on the topic of this paper, as well as anonymous reviewers of earlier versions of this work for their valuable and constructive comments.
}

\nolinenumbers 

\EventEditors{John Q. Open and Joan R. Access}
\EventNoEds{2}
\EventLongTitle{42nd Conference on Very Important Topics (CVIT 2016)}
\EventShortTitle{CVIT 2016}
\EventAcronym{CVIT}
\EventYear{2016}
\EventDate{December 24--27, 2016}
\EventLocation{Little Whinging, United Kingdom}
\EventLogo{}
\SeriesVolume{42}
\ArticleNo{23}

\begin{document}

\maketitle

\begin{abstract}
In the preprocessing framework for dealing with data uncertainty, one is given a set of regions that each represent one uncertain data point.
Then, one is allowed to preprocess these regions to create some auxiliary structure, such that when a realization of the regions is given, this auxiliary structure can be used to reconstruct some desired output geometric structure more efficiently than would have been possible without it. 
A classic result in this framework is that a set of $n$ unit disks of constant ply can be preprocessed in $O(n\log n)$ time such that the convex hull of any realization can be reconstructed in $O(n)$ time. 
In this work we revisit this problem and show that, in fact, we can reconstruct the convex hull in time proportional to the number of \emph{unstable} disks, which may be sublinear, and that such a running time is the best possible. A disk is \emph{stable} if the combinatorial structure of the convex hull does not depend on the location of its realized point. We also extend our results to overlapping disks, and disks of varying radii.

The main tool by which we achieve our results is by using a \emph{supersequence} as the auxiliary structure constructed in the preprocessing phase.
That is, we output a sequence of the disks (possibly with repetitions) such that the convex hull of any realization is guaranteed to be a subsequence of this sequence. 
Using a supersequence as the auxiliary structure allows us to decouple the preprocessing phase from the reconstruction phase in a stronger sense than was possible in previous work, resulting in two separate algorithmic problems which may be independent interest. 
\end{abstract}

\clearpage\tableofcontents

\clearpage
\section{Introduction}

\subparagraph*{Preprocessing framework.} 
The {\em preprocessing framework} for dealing with data uncertainty in geometric algorithms was initially proposed by Held and Mitchell \cite{held2008triangulating}. In this framework, we have a set $\RR = \{R_1, R_2, \ldots, R_n\}$ of {\em regions}, often in $\R^2$, and a point set $P = \{p_1, p_2, \ldots, p_n\}$ with $p_i \in R_i$
(we also write $P \from \RR$).
This model has 2 consecutive phases: a preprocessing phase, followed by a reconstruction phase. In the preprocessing phase we have access only to $\RR$ and we typically want to preprocess $\RR$ in $O(n \log n)$ time to create some linear-size auxiliary data structure which we will denote by $\Xi$. In the reconstruction phase, we have access to $P$ and we want to construct a desired output structure $S(P)$ on $P$ using $\Xi$ faster than would be possible otherwise (see Figure~\ref {fig:intro-prep}).
L{\"o}ffler and Snoeyink~\cite{loffler2010delaunay} were the first to use this model as a way to deal with data uncertainty: one may interpret the regions $\RR$ as {\em imprecise} points, and the points in $P$ as their true (initially unknown) locations. This interpretation of the preprocessing framework
has since been successfully applied to various problems in computational geometry~\cite {buchin2011delaunay,devillers2011delaunay,ezra2013convex, loffler2013unions,van2010preprocessing,hkls-pippf-22}.

\begin {figure}
 \begin{center}
  \includegraphics[width=\textwidth]{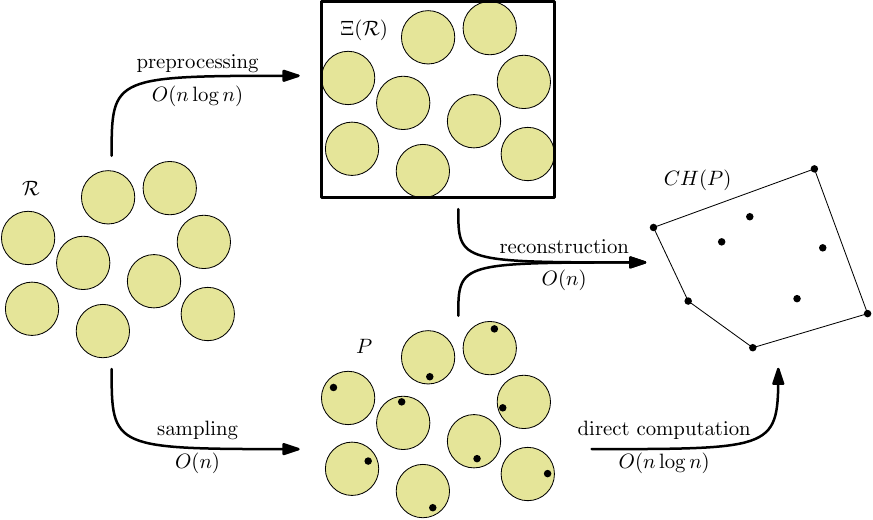}
 \end{center}
 \caption {A set of disjoint unit disks $\RR$ can be preprocessed into an auxiliery structure $\Xi(\RR)$ in $O(n\log n)$ time, such that the convex hull of a set of points $P$ that respects $\RR$ can be computed in linear time using $\Xi(\RR)$ (compared to $\Theta(n \log n)$ time without preprocessing)~\cite {held2008triangulating,van2010preprocessing,BLMM11}.}
 \label {fig:intro-prep}
\end {figure}

\subparagraph*{Convex hull.} 
The {\em convex hull} of a set of points in $\R^2$ is the smallest convex set that contains all points.
The algorithmic problem of computing the convex hull is arguably one of the most fundamental problems in computational geometry~\cite {ch-survey}, and has been considered in an impressive range of different models for geometric uncertainty~\cite {gs-spscp-90,gss-cscah-93,nagai1999convex,edalat2001,hermann2007robust,loffler2010largest,evans2011possible,ezra2013convex,joskowicztopological,devillers2016smoothed,DBLP:journals/algorithmica/AgarwalHSYZ17,DBLP:conf/cccg/HuangR20} (refer to Appendix~\ref{sec:rwork}).
The first result for convex hulls within the preprocessing framework can in fact be derived from the work by Held and Mitchell~\cite {held2008triangulating}, who show that a set of disjoint unit disks can be preprocessed in $O(n \log n)$ time such that a triangulation can be reconstructed in linear time: since any triangulation contains at least the edges of the convex hull, the convex hull can be extracted in the same amount of time (see Figure~\ref {fig:intro-prep}).
Similarly, the work by van Kreveld~\etal~\cite {van2010preprocessing} implies that a set of disjoint disks of arbitrary radii (or a set of disjoint polygons of constant complexity) can also be preprocessed in $O(n \log n)$ time for the same purpose, and the work by Buchin~\etal~\cite {BLMM11} implies that the same is true for a set of moderately overlapping unit disks. While these results focus on different problems (triangulation or Delaunay triangulation) and the convex hull is ``only a by-product'', Ezra and Mulzer~\cite {ezra2013convex} explicitly study the convex hull and show that any set of {\em lines} in the plane can also be preprocessed to speed up the reconstruction of the convex hull, although not to linear time: they achieve a reconstruction time of $O(n \alpha(n) \log^* n)$ in expectation.

\subparagraph* {Sublinear reconstruction.}
Traditionally, the aim in the preprocessing model has been to achieve reconstruction times that are faster than computing a solution from scratch, with the understanding that there is a natural lower bound of $\Omega(n)$ time to reconstruct $S(P)$, since even replacing each region of $\RR$ with the corresponding point in $P$ will take linear time.
However, for problems in which the output is an ordered sequence of points, this reasoning is not entirely satisfactory, for two reasons:
\begin {enumerate}
  \item it is possible that not all elements of $P$ appear in $S(P)$; in this case it is not a priori clear that we need to spend time retrieving points that do not need to be output;
  \item even for those points that do appear in $S(P)$, in some applications one might be happy knowing just the order and not the exact points; if this order is already clear from $\RR$ we also do not necessarily need to spend even $|S(P)|$ time in the reconstruction phase.
\end {enumerate}
Motivated by these observations, van der Hoog \etal~\cite {hkls-paip-19, hkls-pippf-22} explore how the preprocessing model may be modified to allow for such sublinear reconstruction; 
for instance, for sorting they achieve reconstruction time proportional to the entropy of the interval graph, which can be less than $n$ if many intervals are isolated. 
In order to support sublinear reconstruction, van der Hoog \etal introduce an additional phase to the preprocessing framework, see Section~\ref{sec:2d-sub} for more details.

\begin {figure}
 \begin{center}
  \includegraphics{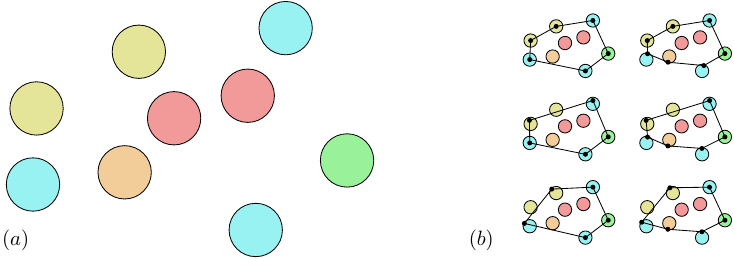}
 \end{center}
 \caption {(a) A set of disjoint unit disks $\RR$. (b) All possible combinatorial convex hulls. We identify five types of disks:
(i) {\em stable impossible interior} disks that never contribute to the convex hull (red);
(ii) {\em unstable potential interior} disks that may or may not contribute (orange);
(ii) {\em unstable potential boundary} disks that may or may not contribute (yellow);
(iii) {\em unstable guaranteed boundary} disks that are guaranteed to contribute, but for which the location may influence the structure of the hull (blue); and
(iv) {\em stable guaranteed boundary} disks that are guaranteed to contribute, and the structure is independent of their location (green).
 (Refer to Section~\ref {sec:classifications} for precise definitions).}
 \label {fig:intro-types}
\end {figure}

\subsection {Contribution}

\subparagraph*{Supersequences.} 
In previous work in the preprocessing framework, the \emph {auxiliary structure}, which is produced during the preprocessing phase, can be anything; and indeed very different techniques and have been used successfully in the past. While this illustrates the breath of the framework, it also makes it hard or impossible to separate the preprocessing phase from the reconstruction phase: they only make sense when viewed as a whole.

In this work, we introduce a new generic auxiliary structure: the {\em supersequence}.
In principle, it is applicable to any computational problem where the output structure $S(P)$ is an ordered subset of $P$; that is, $S(P) = \langle p_{i_1}, p_{i_2}, \ldots, p_{i_s} \rangle$ for some indices $\{i_1, i_2, \ldots, i_s\} \subseteq [n]$.
For such problems, an attractive option to use for the auxiliary structure $\Xi$ is a \emph {supersequence} of the regions corresponding to $S$; that is, $\Xi$ is a sequence of (possibly reccuring) elements of $\RR$ with the guarantee that, no matter where the true points $P$ lie in their regions, the sequence of elements of $P$ which we would obtain by replacing the regions in $\RR$ by their points will always contain $S(P)$ as a subsequence.

A supersequence is arguably a more meaningful structure than the structures used in previous work; as a result, algorithms for the two phases may also be of independent interest, outside of the preprocessing framework.
In addition, the supersequence viewpoint naturally allows us to incorporate sublinear reconstruction times in certain settings, something that was not possible in the earlier work and has only been explored more recently.
We believe the concept is potentially applicable to many other geometric problems (for example, the Pareto front, the TSP tour, or the shortest path in a polygon).


This idea is somewhat related to the concept of 
{\em heriditary} algorithms, introduced by Chazelle and Mulzer~\cite {hereditary09}, in which 
one is given a structure 
on a superset, and is interested in computing the structure on a subset.
Heriditary algorithms have been used in the context of the preprocessing framework, for instance, van Kreveld~\etal~\cite {van2010preprocessing} in their reconstruction phase first produce  a triangulation of red and blue points, where the red points are the points of interest, and then describe a linear-time hereditary algorithm for triangulations to obtain the final result.
Our approach of using supersequences is different in that we do not require explicit construction of a structure on a superset; instead, in the reconstruction phase we can trivially replace the disks by their points and are left with a pure computational problem (Open Problem~\ref {prob:2d}, in the case of convex hulls).

\subparagraph*{Convex hull supersequences.}
In this work, we focus on the \emph {convex hull} as our structure of interest; that is, $S(P)$ is the sequence of vertices of the convex hull,\footnote {For technical reasons we actually compute four {\em quarter hulls} separately; see Section~\ref {sec:quarter-hulls}} and $\Xi(P)$ is a sequence of disks (see Figure~\ref {fig:intro-supseq}). 
It is worth noting that the notion of a supersequence is flexible enough to allow {\em some} disks to appear in the sequence very often, or even a linear number of times (which in fact would be necessary when dealing with arbitrary disks rather than unit disks), as long as the total sequence length is not too large.

In order to support sublinear reconstruction, we must identify subsets of disks that are guaranteed to not take part in the convex hull, and subsets of disks that are guaranteed to take part in the convex hull in a particular order (see Figure~\ref {fig:intro-types}).
In Section~\ref {sec:classifications} we formalize these notions and discuss how they relate to other properties of set of disks that have been used in the past, as well as discuss computational aspects.

We study several variants of the problem: input regions are disjoint unit disks, partially overlapping disks, or disks of varying sizes. 
Our results are formally stated in Section~\ref {sec:results}.

\begin {figure}
 \begin{center}
  \includegraphics{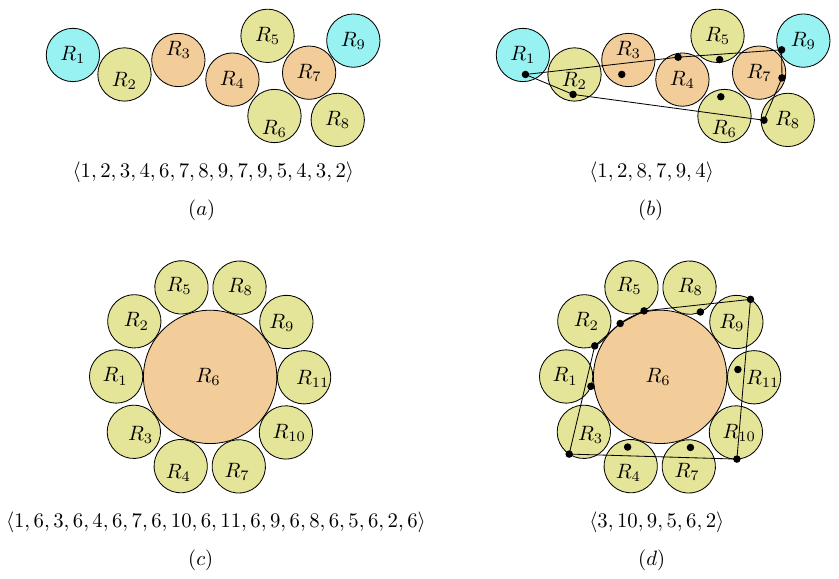}
 \end{center}
 \caption {(a) A set of disjoint unit disks $\RR$ and a sequence of disks (indices) that is guaranteed to contain the vertices of the convex hull in the correct order. (b) A possible true convex hull, and the corresponding subsequence. (c-d) An example with disjoint disks of different radii; note that $R_6$ must appear in the sequence a linear number of times.}
 \label {fig:intro-supseq}
\end {figure}


\subsection {Limitations \& Open Problem}
\label{sec:limits}

We now discuss some limitations and an open problem. Note that we assume we are working in the real RAM model of computation, which is consistent with earlier work in our framework. 

\subparagraph* {Restrictions on the supersequences.}
Our reconstruction algorithms require not only that the auxiliary structure is a supersequence of the desired output, but we also require that these sequences are {\em \nice}. 
This property is naturally fulfilled by our preprocessing algorithms and simplifies the reconstruction task.
However, it is not clear whether this restriction is necessary for the reconstruction problem, and being able to reconstruct the convex hull from {\em any} supersequence would even further decouple the preprocessing and reconstruction phases.

\begin {oproblem} \label {prob:2d}
  We are given a sequence of points $P$ in $\R^2$ (possibly with duplicates), and the guarantee that there exists a subsequence $Q$ of $P$ such that the convex hull of $P$ equals the convex hull of $Q$, and the points in $Q$ are sorted in counterclockwise order.
  Is it possible to compute the convex hull of $P$ in $o(n \log n)$ time?
\end {oproblem}

Note that, although we require handling duplicates in our application, Problem~\ref {prob:2d} is also open even when all points are unique.
We believe this question is of independent interest.

\subparagraph* {General position.}
For ease of exposition, throughout we assume {\em general position} of the disks in the preprocessing phase and general position of the true points in the reconstruction phase.
As mentioned in Section~\ref{sec:classifications}, 
general position for disks means that disks have non-zero radius, no two disk centers coincide, and no three disks are tangent to a common line.
For points, general position means all points are distinct and no three points are collinear.

However, because we will actively duplicate regions and points in our sequences, we will need to drop (part of) this assumption in some subroutines. 
We expect our general position assumptions may be completely lifted by using
symbolic perturbation \cite{em-ss-90}; however, this will certainly complicate the discussion by introducing additional cases in several arguments.


\subsection {Organization}

The remainder of this paper is organized as follows.

In Section~\ref{sec:classifications} we first discuss how uncertainty regions modeled as disks can be classified into types based on their potential roles in the convex hull, which is of independent interest outside the context of the preprocessing framework.

Then in Section~\ref{sec:2d-prelim} we formally define the problem and state our result for preprocessing disks of bounded ply and radii in $\R^2$ into a supersequence which can be used for efficiently reconstructing the convex hull. 

In Section 4, we discuss how to preprocess our set of disks. We first discuss some preliminaries, and then consider disjoint unit disks, overlapping unit disks, and finally disks of varying (though bounded) radii.
 
In Section~\ref {sec:2d-rec} we discuss the reconstruction phase. First, in Section~\ref{sec:2d-rec-standard} we show that it is always possible to recover the convex hull from a {\em \nice} supersequence. Whether it is possible for general supersequences is left as an open problem. Finally, in Section~\ref{sec:2d-sub} we show how to use our classification from Section~\ref{sec:classifications} to reconstruct the convex hull in time proportional to the number of {\em unstable} disks.

In Section~\ref {sec:conclusion}, we summarize our findings and discuss directions for future research.


\section{Classification of Uncertainty Disks for Convex Hulls} \label {sec:classifications}

In this section we discuss a {\em classification} of uncertainty regions modeled as disks into five types, depending on their potential roles in the convex hull.
The results in this section are in principle independent of the preprocessing framework.

Throughout, for a subset $X$ of objects in the plane (i.e.\ regions or points), we write $\CH(X)$ to denote the convex hull of $X$.

\begin {figure}
 \begin{center}
  \includegraphics{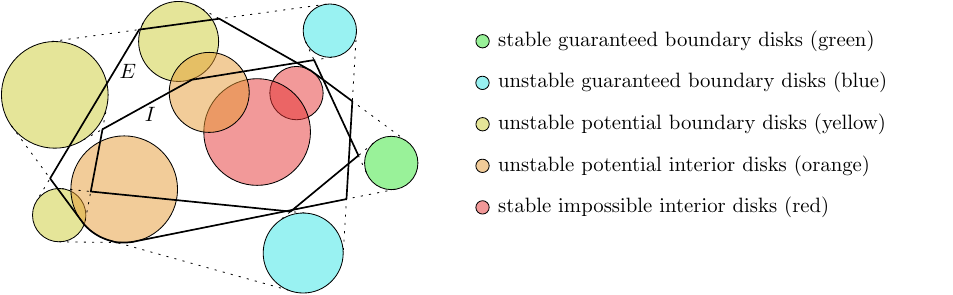}
 \end{center}
 \caption {A set of disks and their classification, and the regions $I$ and $E$.}
 \label {fig:EI}
\end {figure}

We may classify disks
in the plane in multiple ways with respect to the convex hull in the preprocessing framework. 
For ease of description, we will assume general position: all disks have positive (nonzero) radius, no two disks have a common center, and no line is co-tangent to three or more disks or disk centers.

First, following e.g.~\cite {bruce2005efficient}, we classify the disks on whether their corresponding points will appear on the convex hull or not.
We will define {\em impossible}, {\em potential}, and {\em guaranteed} disks.

\begin {definition}\label{def:guaranteed}
  For a set of disks $\RR$ in the plane, we classify each disk $D\in \RR $ as one of the following three types: $D$ is either
  \begin{itemize}
      \item \emph{impossible} when its realization is never a vertex of the convex hull of any realization of $\RR$
      \item \emph{guaranteed}\! when its\! realization is always a vertex of the\! convex hull of any realization of $\RR$
      \item \emph{potential} when it is neither impossible nor guaranteed.
  \end{itemize}
 \end {definition}

Prior work has considered the union or intersection of all possible realizations of $\RR$ under different names (see for example \cite{nyt-chpii-98,Sember2011}), and these structures are  closely related to our above classification of disks. In particular, the above definitions can be equivalently phrased in terms of inclusion or exclusion in certain regions; the following observations have been made before.
Let $I$ be the intersection of all halfplanes that intersect all disks. 
Let $E$ be the intersection of all halfplanes that fully contain at least $n-1$ disks. See Figure~\ref{fig:EI}. 
\begin{observation}\label{obs:polygons}
A disk is impossible if and only if it is contained in $I$, and a disk is guaranteed if and only if it lies outside $E$.  
\end{observation}

Second, we may also classify the disks according to whether the location of their corresponding point influences the combinatorial structure of the convex hull or not.

\begin {definition}\label{def:stable}
  A disk $D \in \RR$ is called {\em stable}, if for any fixed realization $P'$ of $\RR'=\RR \setminus \{D\}$, the combinatorial structure\footnote{The ordered list of disks corresponding to the clockwise sorted order of the vertices of the convex hull.} of the convex hull of  $P'\cup\{p\}$ is the same for any realization $p$ of $D$. 
  A disk that is not stable is \emph{unstable}.
\end {definition}

Note that stable disks are candidates for saving on reconstruction time, as they need not necessarily be inspected during the reconstruction phase to recover the convex hull. So, the best we may hope for is a reconstruction time proportional to the number of unstable disks.

Third, we classify disks depending on whether they appear on the convex hull of the disks themselves.

\begin {definition}\label{def:boundary}
  A disk $D \in \RR$ is called a {\em boundary} disk if the boundary of $\CH(\RR)$ intersects $D$; otherwise $D$ is called an {\em interior} disk. 
\end {definition}


\begin {figure}
 \begin{center}
  \includegraphics{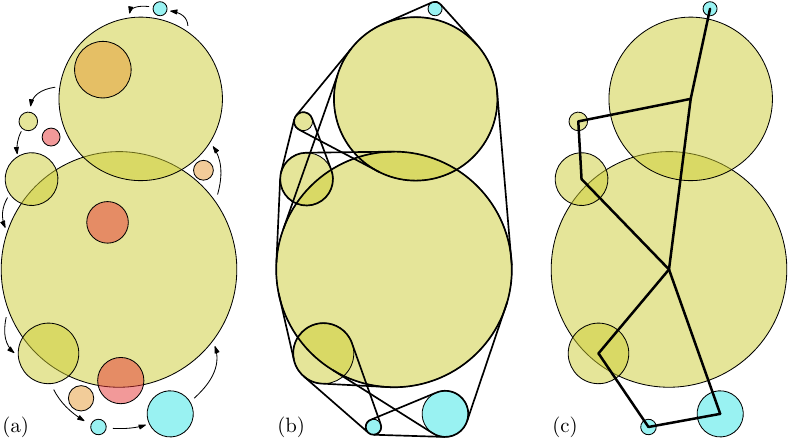}
 \end{center}
 \caption {(a) The subset of {\em boundary} disks in $\RR$ defines a sequence $\cal B(R)$ as illustrated with arrows. In this example, two disks appear twice in the sequence.
 (b) The disks in $\cal B(R)$ form a sequence of {\em strips}, that together form the {\em strip hull}. (c) The {\em spines} of the strips form the {\em spine hull}.}
 \label {fig:strips-and-spines}
\end {figure}
Let $\mathcal{B}(\RR)=\langle B_1,\ldots,B_k\rangle$ be the sequence of boundary disks, indexed in counterclockwise order of appearance on the convex hull, see Figure~\ref {fig:strips-and-spines}(a). Note that not every disk in $\RR$ needs to appear in $\mathcal{B}(\RR)$, and furthermore, for non-unit disks, the same disk may appear multiple times, that is we may have $B_i=B_j$ for $i\neq j$. It is known that $k\leq 2n-1$ \cite{r-chad-91}.
For any $i \in [1,k]$, define the $i$th \emph{strip} as $s_i=\CH(B_i,B_{i+1})$ (where $s_k=\CH(B_k,B_1)$). We refer to the chain of strips, ordered by increasing $i$, as the \emph{strip hull} of $\RR$. (Indeed the convex hull of all strips is equal to that of all regions, where straight segments on the boundary follow individual strips and turns occur at the intersection of the common disk with the next strip.) Furthermore, for any $i$, we call the line segment connecting the centers of $B_i$ and $B_{i+1}$ the {\em spine} of the corresponding strip. Observe that the spines themselves define a plane geometric graph, which we call the {\em spine hull}, see Figure~\ref {fig:strips-and-spines}(c). 
Note that when $\RR$ is a set of unit disks, the spine hull is simply the convex hull of the disk centers.


The classifications in Definitions \ref{def:guaranteed}, \ref{def:stable}, and \ref{def:boundary}, in principle give $3 \times 2 \times 2 = 12$ different types of disks. However, due to certain dependencies many of these types are not possible.

\begin{observation} \label {obs:implications}
We observe the following implications:
\begin {itemize}
  \item every guaranteed disk is a boundary disk;
  \item every potential disk is unstable; and
  \item every impossible disk is both interior and stable.
\end {itemize}
\end{observation}

\begin {proof}
First, observe that a guaranteed disk $D\in \RR$ must be a boundary disk, as by Observation~\ref{obs:polygons}, $D$ must lie outside $E$, or equivalently outside the convex hull of $\RR\setminus \{D\}$. 

Second, a potential disk $D\in\RR$ must be unstable, as $D$ being potential means there is both a realization of $\RR$ where $D$ is realized as a vertex of the convex hull and a realization of $\RR$ where it is not a vertex of the convex hull, and thus the combinatorial structure depends on the realization of $D$.%
\footnote{The definition of stable does not directly prohibit a scenario where for some fixed realization of $\RR\setminus \{D\}$, any realization of $D$ is always a vertex of the convex hull, while for some other fixed realization of $\RR\setminus\{D\}$, any realization of $D$ is never a vertex of the convex hull. However, such a scenario is impossible due to the continuity of our disk based uncertain regions.}

Finally, an impossible disk must be both interior and stable. Clearly it is stable as its realization is never a vertex of the convex hull and thus does not affect its combinatorial structure. It is interior as by Observation~\ref{obs:polygons} it lies inside $I$ which lies strictly inside the conex hull of $\RR$.
\end{proof}

Observation~\ref{obs:implications} implies there are only 5 possible types of disks (see 
Figure~\ref{fig:EI}).
%
%
\subsection {Computation of the Disk Classification}
\label{sec:computeclass}

Here we argue that the disk classifications into the 5 possible types from 
Section~\ref{sec:classifications} can be computed in $O(n \log n)$ time. 
First, we will narrow down the possible configurations of disks in which {\em stable guaranteed boundary disks} can appear.

\begin {observation}\label{obs:once}
  A guaranteed disk can appear at most once on the strip hull.
\end {observation}

\begin {proof}
  Suppose a guaranteed disk $D$ appears twice on the strip hull. Then there must be two other disks $A$ and $B$ such that there are four points $a \in A$, $b \in B$, and $d_1, d_2 \in D$ on the convex hull of $\RR$ belonging to the disks in the order $a,d_1,b,d_2$. Let $m$ be the intersection point of $ab$ with $d_1d_2$. If we place the point of $A$ at $a$, the point of $B$ at $b$, and the point of $D$ at $m$, $D$ does not contribute a point to the convex hull, a contradiction.
\end {proof}

\begin{lemma}\label{lem:stable-vs-instable}
  Let $D$ be a guaranteed boundary disk.
  Then $D$ is stable if and only if both of its adjacent strips on the strip hull are empty, and both neighbors of $D$ on the strip hull are guaranteed.
\end{lemma}

\begin {proof}
  Let $D$ be a guaranteed boundary disks.
  By Observation~\ref{obs:once}, $D = B_i$ for a unique $i$.

  We will first argue the $\Leftarrow$ direction.
  Suppose $B_{i-1}$ and $B_{i+1}$ are both guaranteed, and the strips $s_{i-1}$ and $s_i$ are empty (that is, do not intersect any other disks).
  By definition of {\em guaranteed}, the three disks $B_{i-1}, B_i, B_{i+1}$ will contribute three points $p_{i-1}, p_i, p_{i+1}$ to the convex hull.
  Furthermore, there can be no other point on the hull between $p_{i-1}$ and $p_i$ or between $p_i$ and $p_{i+1}$, since such a point would need to lie in $s_{i-1}$ or $s_i$, which is impossible.
  That is, $p_{i-1}$, $p_i$ and $p_{i+1}$ appear {\em consecutively} on the hull for all possible realizations, which implies that $p_i$ is indeed stable.

  Next,  we will argue the $\Rightarrow$ direction.
  Suppose that $B_i$ is stable.
  Suppose further for contradiction that $B_{i+1}$ is {\em not} guaranteed (the argument for $B_{i-1}$ is symmetric).
  Then there are two realizations $P$ and $P'$ such that $p_{i+1}$ is a hull vertex in $P$, but not in $P'$. Since disks are continuous connected regions, we may continuously move the points from $P$ to $P'$, and we must reach some intermediate realization $P''$ in which $p_{i+1}$ is exactly {\em on} the convex hull boundary. Now, moving $p_i$ inside a small neighbourhood will make $p_{i+1}$ be on or not on the convex hull, contradicting that $B_i$ is stable.
  Similarly, assume for contradiction that $s_i$ is not empty (the argument for $s_{i-1}$ is symmetric). Then there exists a realization $P$ with a point $p_j$ inside $s_i$; draw a line through $p_j$ that intersects both $B_i$ and $B_{i+1}$ and place $p_i$ and $p_{i+1}$ on this line. We may assume $p_i$, $p_j$, and $p_{i+1}$ are all on the convex hull (if not, there must be another point in $s_i$ and we can move $\ell$ parallel to itself until it intersects this point; repeat if necessary). But then again, moving $p_i$ inside an epsilon neighbourhood will cause $p_j$ to be on the hull or not, contradicting that $B_i$ is stable.
\end {proof}

Now we are ready to prove the following.

\begin{lemma}\label{lem:compute-classification}
Given a set $\RR$ of $n$ disks, we can classify them in $O(n \log n)$ time.
\end{lemma}

\begin{proof} 
First, the boundary disks can be easily detected in $O(n \log n)$ time by computing the convex hull of the disks~\cite {Devillers1995IncrementalAF}.
Second, Observation~\ref {obs:polygons} allows us to classify the disks into guaranteed, potential, and impossible disks using a simple $O(n \log n)$ algorithm~\cite {nyt-chpii-98}.
By Observation~\ref {obs:implications} all impossible interior disks are stable. 
It remains to distinguish the stable from the unstable guaranteed boundary disks.

By Lemma~\ref {lem:stable-vs-instable}, in order for any guaranteed boundary disk $B_i$ to be stable, its adjacent strips must be empty and its neighbors on the strip hull must be guaranteed. 
As for the running time for finding the stable guaranteed boundary disk, as discussed above we can already determine all boundary disks, and for each such disks whether it and its neighbors are guaranteed, in $O(n\log n)$ time. Thus we only determine whether the adajacent strips are empty, though this can be done for all stips in $O(n\log n)$ time using  Lemma~\ref{lem:comprri}.\footnote{Technically that lemma assumes unit disks, though the same arguments hold for non-unit disks. Namely, the strips are still pseudodisks as the spines form a plane geometric graph.}
\maarten {A bit weird to refer so far forward, but ok.}
\end{proof}

\section{The Preprocessing Framework: Definitions, Preliminaries, and Results}
\label {sec:2d-prelim}

In this section we formally define the problem and state our result for preprocessing unit disks of bounded ply in $\R^2$ for reconstructing the {\em convex hull}.
While our focus is on unit disks in this paper, many of the definitions and observations in this section hold more generally for any set of disks.
Before we can state our results, we need to first discuss 
a notion of {\em quarter hull} to decompose the problem into smaller parts, 
and a notion of {\em smoothness}.


\subsection{Quarter Hulls} \label {sec:quarter-hulls}

Note that the convex hull is cyclic in nature: it has no well-defined start or end point, 
and in the presence of uncertainty, it may be difficult even to designate an arbitrary disk as the start point, as there may not be any guaranteed disks.

To make life easier, we opt to focus on only one quarter of the convex hull: the part from the rightmost point on the hull to the topmost point on the hull (i.e. the positive quadrant).
Indeed, the idea of splitting the computation of the convex hull into separate pieces is well established~\cite {ANDREW1979216}, and choosing quarters will have the added benefit that within such a quarter, no vertex can have an internal angle smaller than $90^\circ$, which will aid our geometric arguments later.

Formally, the {\em quarter hull} of a point set $P$ is the sequence of points on the convex hull starting from the rightmost point in $P$ and following the hull in  counterclockwise order until ending with the topmost point in $P$.

We will define a supersequence for one quadrant; clearly, to obtain a supersequence for the entire hull, one can simply repeat this process four times for appropriately rotated copies of $\RR$.
Note that in subsequent sections, depending on the context, the term \emph{convex hull} may refer to either the quarter hull or the full convex hull. \maarten {ideally, remove this sentence, assuming we have fixed all occurences}

In the remainder of this work, we also adapt the notation established in Section~\ref {sec:classifications} to the notion of quarter hulls.
That is, a disk will be called {\em guaranteed} or {\em impossible} when it is guaranteed or impossible to contribute to the quarter hull,
and it will be {\em stable} when it does not influence the combinatorial structure of the quarter hull. Refer to Figure~\ref {fig:quarter-hull-types}.

\begin {figure}
 \begin{center}
  \includegraphics{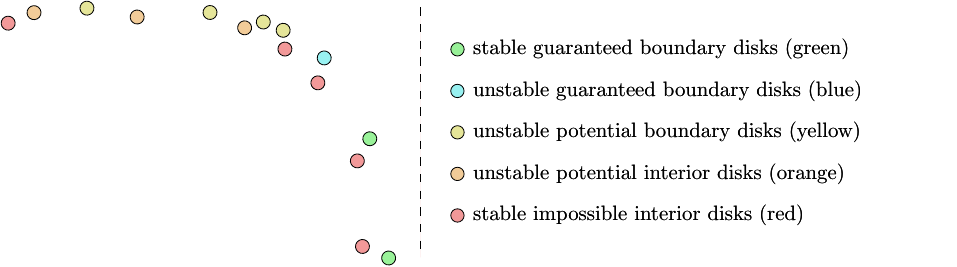}
 \end{center}
 \caption {A set of disjoint unit disks and their types with respect to the top right quarter hull.}
 \label {fig:quarter-hull-types}
\end {figure}

Stable disks for the convex hull are not necessarily stable for the quarter hull, but not vice versa.
However, observe that the total number of unstable disks for the entire convex hull is at most a constant number more than the sum of the number of unstable disks in the four quarter hull problems; hence, dividing the problem into quarters does not jeopardize our aim to achieve sublinear reconstruction times.

\begin {observation} \label {obs:unstable-quarter}
  Let $\RR$ be a set of disks in the plane. Let $m$ be the total number of unstable disks in $\RR$ for the convex hull.
  Let $m_1, m_2, m_3, m_4$ be the numbers of unstable disks in $\RR$ for the four quarter hulls.
  Then $m < m_1 + m_2 + m_3 + m_4 < 4m$.
\end {observation}
\maarten {For the full version, maybe include a proof?
\begin{proof}
  Proof should be based on these observations:
  - Every disk that is stable for the convex hull is also stable in at least one of the four quarter hulls (either it is in the interior, or if it is on the boundary, just look at the other size of the hull).
  - It could be that there are unstable disks for the quarter hull that are stable for the convex hull, but then they must be stable for at least one other quarter hull.
  - ...
\end{proof}}

\subsection {Supersequences} \label {sec:supersequences}

We are now ready to formally define supersequences.

\begin {definition}
  Let $\RR$ be a set of $n$ regions in the plane.
  A {\em quarter-hull-supersequence} of $\RR$ is a sequence $\Xi$ of (possibly reoccurring) regions of $\RR$
  such that,
  for any point set $P \from \RR$, 
  the quarter hull of $P$
  is a subsequence of $\Xi$.
  
\end {definition}

We will focus on the case where $\RR$ is a set of unit disks of bounded ply $\Delta$ in this work, where the {\em ply} of a set of regions is the maximum number of regions that any point in the plane is contained in.
Replacing the disks in a quarter-hull-supersequence by their points, we now obtain a sequence of points that contains all points of the quarter hull of $P$ in the correct order. Since it is not clear if the convex hull of any such sequence can be computed in linear time (Open Problem~\ref {prob:2d}), we additionally introduce {\em \nice} sequences. 

For two points $p, q \in \R^2$, let $\nwarrow(p,q)$ be the signed distance from $p$ to $q$ projected on a line with slope $-1$; that is, 
\[\nwarrow(p,q) = (p.x - p.y)-(q.x - q.y).\]

\begin {definition}
  A sequence $X$ of points in $\R^2$ is {\em $(\alpha, \beta)$-\nice} if:
  \begin {itemize}
    \item for any two elements $p_i$ and $p_j$ in $X$ with $i < j$ we have
      that if $p_i$ is on the quarter hull of $X$, then $\nwarrow(p_j,p_i) \le \alpha$, and symmetrically if $p_j$ is on the quarter hull of $X$, then $\nwarrow(p_i,p_j) \ge -\alpha$.
    \hfill (distance property)
    \item for any disk $D$ in the plane of diameter $d>1$, there are only $\beta d$ distinct points from the sequence in $D$.
    \hfill (packing property)
  \end {itemize}
\end {definition}

Essentially, a \nice supersequence may have points that are out of order, but the extent in which they can be out of order is limited by the parameters $\alpha$ and $\beta$, which makes the reconstruction problem significantly easier.

We will also say a quarter-hull-supersequence $\Xi$ is {\em $(\alpha, \beta)$-\nice} if the corresponding sequence of points is {\em $(\alpha, \beta)$-\nice} for any $P \from \RR$.

\subsection{Problem Statement \& Results}
\label{sec:results}

We now state our general problem of preprocessing disks in the plane for the convex hull, and then afterwards state our results. 

\maarten {for conference version, maybe we don't need the problem statement?}

\begin{problem}\label{prob:general}
Let $\RR$ be a set of $n$ disks. Can we construct a supersequence $\Xi$ of $\RR$ such that given any $P\from \RR$ we can construct $\CH(P)$ in $O(n)$ time by using $\Xi$?
\end{problem}

While the above problem is stated for a general set of disks, some restrictions will be required. For instance, if we allow $n$ copies of the same disk, then constructing such a sequence $\Xi$ is not possible. 

We now state our main results for preprocessing unit disks in the plane of bounded ply for the convex hull. Our first two theorems imply that the classical result of $O(n \log n)$ preprocessing for $O(n)$ reconstruction (as in Figure~\ref {fig:intro-prep}) can indeed be replicated for unit disks of bounded ply using supersequences.

\begin {theorem} \label{thm:main-general-unspecific}
  Let $\RR$ be a set of $n$ disks in the plane of constant ply and a constant range of sizes.
  There exists a $(O(1),O(1))$-\nice quarter-hull-supersequence $\Xi$ of $\RR$ of size $O(n)$,
  and it can be computed in $O(n \log n)$ time.
\end {theorem}

Note the following theorems do not require the ply or size range to be constant.

\begin {theorem} \label {thm:mainunitdisks-reconstruction}
  Let $\RR$ be a set of $n$ disks.
  Given an $(\alpha,\beta)$-\nice quarter-hull-supersequence $\Xi$ of $\RR$,
  and given a point set $P \from \RR$, the quarter hull of $P$ can be computed in $O(\alpha\beta|\Xi|)$ time. 
\end {theorem}

The next theorem implies that we can also obtain sublinear reconstruction using this method, which is a benefit of our approach of directly studying the convex hull, as opposed to the convex hull being product of some other geometric structure of study.

\begin{theorem}\label {thm:sublinear-reconstruction}
 Let $\RR$ be a set of $n$ disks. We are given an $(\alpha,\beta)$-\nice quarter-hull-supersequence $\Xi$ of $\RR$ in which all its stable disks have been marked, and let $\mu$ denote the number of unstable disks in $\Xi$. Then given a point set $P \from \RR$, the quarter hull of $P$ can be computed in $O(\alpha\beta \mu)$ time. 
\end{theorem}

Note for the above theorem, that by nature of sublinear reconstruction we cannot spend time proportional to the output size, therefore the algorithm will output a pointer to the head of a linked list that contains all the elements of the quarter hull in the correct order. See Section~\ref{sec:2d-sub} for details.

When $\RR$ has constant ply and constant size range,  Theorem~\ref{thm:main-general-unspecific},
Lemma~\ref{lem:compute-classification},
Theorem~\ref{thm:sublinear-reconstruction}, and
Observation~\ref{obs:unstable-quarter}
together imply the following corollary.

\begin{corollary}
Let $\RR$ be a set of $n$ disks in $\R^2$ with constant ply and constant size range.
We can preprocess $\RR$ in $O(n \log n)$ time into a supersequence with its stable disks marked, 
from which the convex hull of $P$ can be computed in $O(m)$ time, where $m$ is the number of unstable disks in $\RR$.
\end {corollary}

\section{The Preprocessing Phase}

\subsection {Preliminaries}

We begin by discussing additional definitions and preliminaries for the preprocessing phase.
\maarten {probably goes to appendix for conference version, and we just put the defs here without subsections}

\subsubsection {Definitions and Notation}
\label{sec:defnot}

Throughout this section, let $\RR=\{R_1,\ldots,R_n\}$ be a set of $n$ disjoint unit disks in the plane, arbitrarily indexed. 
We recall some definitions from Section~\ref {sec:classifications}.
Let $\mathcal{B}(\RR)=\{B_1,\ldots,B_k\}$ be the subset of boundary disks (Definition~\ref{def:boundary}), indexed in counterclockwise order of appearance on the convex hull, see Figure~\ref{fig:definitions}.
When $\RR$ is clear from the context, we may write simply $\mathcal B$.

\begin {figure}[t]\centering
\includegraphics {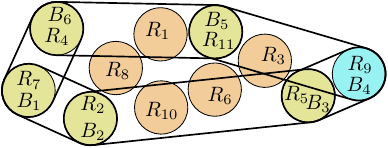}
\caption {The convex hull of $\RR = \{R_1, \ldots, R_{11}\}$, with {\em yellow} and {\em blue} disks on the boundary (depending on whether potential or guaranteed, see Figure~\ref{fig:intro-types}),  {\em orange} internal disks, and {\em strips} consisting of two consecutive boundary disks. In this case $\mathcal{B}(\RR) = \{B_1, \ldots, B_6\} = \{R_7, R_2, R_5, R_9, R_{11}, R_4\}$.}
 \label{fig:definitions}
\end {figure}

For any $i \in [1,k]$, define the $i$th \emph{strip} as $s_i=\CH(B_i,B_{i+1})$ (where $s_k=\CH(B_k,B_1)$). We refer to the chain of strips, ordered by increasing $i$, as the \emph{strip hull} of $\RR$. 
Furthermore, for any $i$, we call the line segment connecting the centers of $B_i$ and $B_{i+1}$ the {\em spine} of the corresponding strip. Observe that the spines themselves define an inner convex polygon\footnote{A stronger requirement than was  the case for potentially non-unit disks in Section~\ref {sec:classifications}.}, which we call the {\em spine hull}. 

For a point set $P \from \RR$, we refer to the vertices  $Q=\{q_1,\ldots q_k\}$ of $P$'s convex hull, as the {\em $Q$-hull}.
\begin {observation} \label {obs:Q-vertices-in-stips}
  Every vertex of the $Q$-hull is contained in at least one strip.
\end {observation}

Since our goal is to recover the {\em quarter hull}, we extend the above hull definitions to quarter hulls. 
The {\em quarter $Q$-hull} is simply the quarter hull of $Q$. The {\em quarter spine hull} is the quarter hull of the centers of the disks in $\mathcal{B}(\RR)$. 
Finally, the {\em quarter strip hull} is the subchain of the strip hull defined by the disks whose centers are on the quarter spine hull, with the addition of two new unbounded {\em half-strips}, one which goes downward from the rightmost disk, and the other which goes leftward from the topmost disk.%
\footnote {Note that the rightmost and topmost extremal disks do not necessarily contribute the rightmost and topmost points in $P$.} 
See Figure~\ref {fig:quarter-hull-definitions}.

\begin {figure}[t] 
\centering
\includegraphics {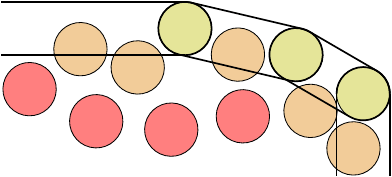}
\caption {The quarter strip hull of $\RR$, with {\em yellow} disks on the boundary,  {\em orange} and {\em red} internal disks (depending on whether unstable potential or stable impossible, see Figure~\ref{fig:intro-types}), and {\em strips} consisting of two consecutive boundary disks; there are two infinite {\em half-strips} from the rightmost disk and the topmost disk.}
\label {fig:quarter-hull-definitions}
\end {figure}

\subsubsection {Geometric Properties}\label{sec:geomprelim}

In this section we prove various geometric properties that will be required for the proof of Theorem~\ref{thm:main-general-unspecific}.

\begin {lemma} \label {lem:intersection_bound}
  Let $P$ be a convex polygon, with edges of length at least $x>0$.
  Let $D$ be a unit radius disk.
  Then $D$ can intersect at most 
  $\lfloor 2 \pi / x \rfloor + 1$ edges of $P$.
\end {lemma}

\begin {proof}[Proof of Lemma~\ref{lem:intersection_bound}]
  Let $D'$ be a disk of radius $1+x$ concentric with $D$.
  Observe that for every edge $e$ of $P$ that intersects $D$, the intersection of $e$ with $D'$ must have length at least $x$.
 
We now argue that the total length of $P$ that can be contained in $D'$ is at most $2 \pi (1 + x)$. More generally, let $\gamma$ be any convex closed curve, bounding a closed convex body $\Gamma$. As $D'$ is convex, the intersection $\Gamma\cap D'$ is also convex, and thus the maximum length of the boundary of $\Gamma\cap D'$ is the length of the circumference of $D'$, namely $2 \pi (1 + x)$. Thus as $\gamma \cap D'$ is a subset of the boundary of $\Gamma\cap D'$, its maximum possible length is $2 \pi (1 + x)$.
  
Combining our upper bound on the length of $P$ contained in $D'$ with our lower bound on the length of edges in $D'$ that intersect $D$, we can conclude that the number of edges of $P$ that $D$ intersects is at most $\lfloor 2 \pi (1 + x) / x \rfloor \leq \lfloor 2 \pi / x \rfloor + 1$.
\end {proof}

\begin {corollary} \label {cor:diskstrips}
  Every unit disk intersects at most 7 strips.
\end {corollary}
\begin {proof}[Proof of Corollary~\ref{cor:diskstrips}]
 A unit radius disk $D$ intersects a strip if and only if a radius 2 disk $D'$ with the same center point intersects the spine of the strip. Now consider the convex polygon $P$ which is the spine hull of the strips. Observe that the edges of $P$ have length at least 2, as the defining disks of a strip are disjoint. Now scale both $P$ and $D'$ by a factor $1/2$, and now the edges of $P$ have length at least 1 and $D'$ is a unit radius disk. So now we can apply Lemma~\ref {lem:intersection_bound} where $x=1$, implying that at most 7 edges of $P$ are intersected by $D'$, which by the above implies that $D$ intersected at most 7 strips.  
\end {proof}

\begin {corollary} \label {cor:ply}
  The arrangement of strips has ply at most 4.
\end {corollary}

\begin {proof}[Proof of Corollary~\ref{cor:ply}]
  Consider any point $z$ in the plane. $z$ is contained in a strip if and only if a unit disk $D$ centered at $z$ intersects the spine of that strip. So consider the convex polygon $P$ which is the spine hull of the strips. Observe that the edges of $P$ have length at least 2, as the defining disks of a strip are disjoint. Thus applying Lemma~\ref {lem:intersection_bound} with $x=2$, we know that $D$ intersects at most 4 edges of $P$, and therefore $z$ intersects at most 4 strips. The ply of the arrangement is the maximum of all points in the plane, and since $z$ was an arbitrary point, the ply is at most 4.
\end {proof}

\begin {figure} 
\centering
\includegraphics {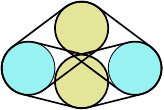}
\caption {An example of four disjoint unit disks where the arrangement of strips has ply 4.}
\label {fig:ply}
\end {figure}

We observe that Corollary~\ref {cor:ply} is in fact tight, as seen by Figure~\ref {fig:ply}.

\subsubsection{Fr\'echet and Hausdorff}

Here we define and discuss basic properties of the Fr\'echet and Hausdorff distance measures. 

Roughly speaking, the Hausdorff distance between two point sets is the distance of the furthest point from its nearest neighbor in the other set.

\begin{definition}
 Given two point sets $P,Q\subseteq \R^2$, their Hausdorff distance is 
 \[\delta_H(P,Q) = \max\{\sup_{p\in P} d(p,Q), \sup_{q\in Q} d(q,P) \}\]
\end{definition}

We use the following definition of Fr\'echet distance from \cite [Definition 3.10]{k-acgp-03}, which applies both to the unit interval $\S^0$, and well as the unit circle $\S^1$.

\begin {definition}[Fr\'echet distance]\label{def:Frechet}
Let $P : \S^k \to \R^2$ and
$Q : \S^k \to \R^2$ 
with $k \in \{0, 1\}$ be (closed) curves. 
Then $\delta_F (P, Q)$ denotes the Fr\'echet
distance between $P$ and $Q$, defined as
\[
  \delta_F (P, Q) := 
  \inf_{\alpha : \S^k \to \S^k, \beta : \S^k \to \S^k}
  \max_{t \in \S^k}
  ||P(\alpha(t)) - Q(\beta(t))||,
\]
where $\alpha$ and $\beta$ range over all continuous and increasing bijections on $\S^k$.
\end {definition}

\cite[Theorem 5.1]{k-acgp-03} discusses the equivalence between Hausdorff distance and Fr\'echet distance for closed convex curves, originally argued in 
\cite {10.1007/BFb0032068}.

\begin {theorem} [\cite {10.1007/BFb0032068}] \label {thm:Hausdorff=Fréchet}
  For two convex closed curves in the plane, the Hausdorff distance is the same as the Fréchet distance.
\end {theorem}

We will require the following lemma, which is a standard observation, though for completeness we include a short proof.
Note that in the following lemma, a convex polygon refers only to the boundary and not the interior of the enclosed region.

\begin {lemma} \label {lem:discrete_Hausdorff}
Given two convex polygons A and B, where every vertex of A is within distance $\varepsilon$ of B and every vertex of B is within distance $\varepsilon$ of A, then the Hausdorff distance between A and B is at most $\varepsilon$.
\end {lemma}
\begin{proof}

Let $p$ be any point on $A$. We wish to argue that $p$ is within distance $\varepsilon$ of some point on $B$. (A symmetric argument can then be applied for any point on $B$). If $p$ is a vertex of $A$ then the property holds by the lemma statement. Otherwise $p$ is an interior point of some segment with endpoints $s$ and $t$, and let $\ell$ be the supporting line of this segment. 

Let $D$ be the $\varepsilon$ radius ball centered at $p$. Assume for contradiction that $D$ does not intersect $B$. Then there are two cases. In the first case $B$ contains $D$ in its interior. This implies $B$ has at least one vertex farther than $\varepsilon$ from $\ell$ on both sides of $\ell$. By convexity $A$ lies entirely on one side of $\ell$, and thus this in turn implies one of these vertices from $B$ is farther than $\varepsilon$ from $A$, which is a contradiction with the lemma statement.

In the second case $B$ is separated from $D$ by some line $m$. Thus $p$ is farther than $\varepsilon$ from $m$, and as $p$ lies on the segment $st$, at least one of the vertices $s$ or $t$ is also separated by $m$ from $B$ and lies farther than $\varepsilon$ from $m$. However, in this case this vertex is also farther than $\varepsilon$ away from $B$, which again is a contradiction with the lemma statement.     
\end{proof}

\subsection{Preprocessing Disjoint Unit Disks}
\label{sec:2d-prep-disjoint}

Our goal in this section is to prove Theorems~\ref{thm:main-general-unspecific} for the case of disjoint unit disks.

\subsubsection{Existence of a Supersequence for Disjoint Unit Disks}
\label {sec:2d-exist}
In this section, we argue that a sequence $\Xi$, as described in Theorem~\ref{thm:main-general-unspecific}, exists.

\begin {theorem} \label{thm:main-disjoint-unitdisks-existence}
  Let $\RR$ be a set of $n$ disjoint unit disks in the plane.
  There exists a $(1,3)$-\nice quarter-hull-supersequence $\Xi$ of $\RR$ of size $O(n)$.
\end {theorem}

We first argue that we can collect the disks that intersect a single strip.

Consider any strip $s$ spanned by two disks $L$ and $R$, and assume w.l.o.g. that the strip is horizontal and is at the top of the strip hull; let $\beta$ be the bottom edge of the strip (the segment connecting the bottommost point of $L$ to the bottommost point of $R$). (See Figure~\ref {fig:order-matters}.) 
Now consider the set $\RR_s \subset \RR \setminus \{L,R\}$ of other disks that intersect $s$. Clearly, all disks in $\RR_s$ must intersect $\beta$. But since they are disjoint, they intersect $\beta$ in a well-defined order.
\begin{lemma} \label {lem:line}
 A portion of the convex hull that intersects $L$ and $R$ may only contain vertices from $\RR_s$ in the order in which they intersect $\beta$.
\end{lemma}
\begin{proof}
Consider any two disks $a$ and $b$ in $\RR_s$ with $a$ left of $b$ in their intersection order with $\beta$. Let $\lambda$ be the bisector of $a$ and $b$. Now assume that $a$ is higher than $b$ (the other case is symmetric). Then $\lambda$ has positive slope, and it intersects $\beta$ (since it passes between $a$ and $b$), therefore $L$ lies completely above $\lambda$. But then we can never have a convex curve that passes through $L$ hit first $b$ and then $a$. So, if $a$ and $b$ both appear on the hull, $a$ comes before $b$.
\end{proof}

Lemma~\ref {lem:line} implies that for a single strip $s$, we can simply put the disks in $\RR_s$ in this order to form a subsequence $\Xi_s$.

\begin {figure}
 \begin{center}
  \includegraphics{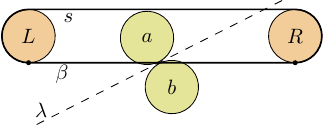}
 \end{center}
 \caption {$a$ and $b$ can only both appear on the upper hull in one order.}
 \label {fig:order-matters}
\end {figure}

Next, we argue that we can simply concatenate the sequences per strip to obtain a single global sequence.
%

\begin{lemma} \label {lem:order}
  Let $P \from \RR$ be a point set, and let $Q = \{q_1, \ldots, q_m\}$ be its convex hull.
  Let $k$ be the number of vertices on the spine hull of $\RR$.
  There exist a mapping $f : [1:m] \to [1:k]$ such that:
  \begin {itemize}
  \item $q_i$ is contained in $s_{f(i)}$; and
  \item if $a,b,c$ are increasing $(\operatorname{mod}  m)$ then $f(a), f(b), f(c)$ are non-decreasing $(\operatorname{mod} k)$
  \end {itemize}
\end{lemma}

\begin{proof}[Proof of Lemma~\ref{lem:order}]
We start by applying Lemma~\ref {lem:discrete_Hausdorff} in order to conclude that the Hausdorff distance between the spine hull and the $Q$-hull is at most 1.
Specifically, Lemma~\ref {lem:discrete_Hausdorff} can be applied since:
\begin{enumerate}[label=(\alph*)]
    \item every vertex of the spine hull is in distance 1 from the $Q$-hull (if not, the corresponding disk would be either completely inside or completely outside the $Q$-hull, both of which are impossible), and
    \item every vertex of the $Q$-hull is in distance 1 from the spine hull, or, in other words, every vertex of the $Q$-hull is contained in one of the strips, 
which follows from Observation~\ref {obs:Q-vertices-in-stips}.
\end{enumerate}
Now given this bound on the Hausdorrf distance, by Theorem~\ref{thm:Hausdorff=Fréchet}, the Fréchet distance between the $Q$-hull and the spine hull is at most 1.

Let this Fréchet mapping be denoted $f$. As it maps the $Q$-hull to the spine-hull, with Fréchet distance $\le 1$, $f$ trivially satisfies the first condition of the lemma. As for the second condition, we consider two cases. First suppose that of three values $f(a)$, $f(b)$, and $f(c)$, that at least two are equal. Then since we are working $\operatorname{mod} k$ this implies the $f(a)$, $f(b)$, $f(c)$ are non-decreasing, regardless of the third value, and so the second condition in the lemma trivially holds. On the other hand, if $f(a)$, $f(b)$, and $f(c)$ are all distinct then the second condition holds due to the monotonicity of the Fréchet distance in Definition~\ref{def:Frechet}.
\end{proof}

\begin {proof} [Proof of Theorem \ref{thm:main-disjoint-unitdisks-existence}]~
Consider the quarter strip hull, as shown in  Figure~\ref {fig:quarter-hull-definitions}.
For each strip (including the two unbounded half-strips), let $\RR_i \subseteq \RR$ be  the set of disks that intersect the strip $s_i$.
By 
Corollary~\ref{cor:diskstrips},
$\sum |\RR_i| \in O(n)$.

For each strip $s_i$ separately, we define a sequence $\Xi_i$ consisting of the disks of $\RR_i$, ordered according to Lemma~\ref {lem:line}. As the disks in $\RR_i$ are disjoint and all intersect the bottom edge of the strip in order, the corresponding points can never be our of order by more than $1$, and no more than $3$ points can be within a region of radius $1$, so $\Xi_i$ is $(1,3)$-\nice.

Finally, to produce the required sequence of the theorem statement, we concatenate the $\Xi_i$ producing $\Xi = \langle \Xi_1, \Xi_2, \ldots, \Xi_k \rangle$. 
To see that $\Xi$ has the desired size, by the above discussion we have 
$|\Xi| = \sum_i |\Xi_i| = O(\sum_i |\RR_i|) = O(n)$.

Let $P \from \RR$ be an arbitrary realization. What remains is to argue that the sequence of points on the quarter hull of $P$ (i.e.\ the quarter $Q$-hull) is a subsequence of $\Xi$.
By Lemma~\ref {lem:order}, there exists a mapping $f$ that maps the points in the $Q$-hull to the strips in the strip hull, such that their cyclic orderings match.

Recall that when defining the quarter strip hull, we included the two unbounded half strips extending to the left of the topmost disk and below the rightmost disk. Let the \emph{bounded quarter strip hull} refer to the portion of the quarter strip hull excluding these two half strips. 
Now consider the image of the quarter $Q$-hull under $f$; this is a sequence of strips of the strip hull. 
Note that not all of these strips need to be strips of the bounded quarter strip hull; however, all points of the quarter $Q$-hull that are mapped to strips not on the bounded quarter strip hull must necessarily lie in the two half strips in the (unbounded) quarter strip hull.

Thus, as $\Xi$ appends the $\Xi_i$ in the ordering of the strips of the quarter strip hull (including the unbounded ones), it suffices to argue that for all $i$, the order of the subset of $Q$-hull vertices mapped to $s_i$ under $f$ appears as a subsequence of $\Xi_i$; however this is guaranteed
by Lemma~\ref {lem:line}.
\end {proof}

\subsubsection{Computation of a Supersequence for Disjoint Unit Disks}
\label {sec:2d-prep}

In this section, we argue that we can also compute such a sequence $\Xi$ in $O(n \log n)$ time.

\begin {theorem} \label{thm:main-disjoint-unitdisks-preprocessing}
  Let $\RR$ be a set of $n$ disjoint unit disks in the plane.
  A $(1,3)$-\nice quarter-hull-supersequence $\Xi$ of $\RR$ of size $O(n)$ can be computed in $O(n \log n)$ time.
\end {theorem}

First, we want to argue that we can compute the arrangement of all disks and strips in $O(n \log n)$ time.

\begin {lemma}
  The arrangement of all strips has complexity $O(n)$
\end {lemma}

\begin {proof}
  By Corollary \ref{cor:ply}, the arrangements of strips has constant ply. 
  Since the strips are in convex position, they form a system of pseudodisks.
  Finally, any system of pseudodisks of bounded ply has complexity $O(n)$ \cite{s-ksacs-91}.
\end {proof}

By the above lemma and Corollary~\ref {cor:diskstrips}, we have that the arrangement of all strips and all disks has complexity $O(n)$.
Furthermore, such an arrangement can be computed in time proportional to its complexity multiplied by an additional log factor~\cite {4567905}.

\begin{lemma} \label {lem:comprri}
  We can compute the sets $\RR_i$, for all $i$ simultaneously, in $O(n\log n)$ time.
\end{lemma}
\begin{proof}
First, we compute the arrangement of all disks and all strips, which can be done in $O(n\log n)$ time as the arrangement has complexity $O(n)$.
Then, for each strip, we identify the set of disks $\RR_i$ that intersect the strip $s_i$.
 Since we constructed the arrangement explicitly, we can just walk through it to collect the disks intersecting each strip.
\end{proof}

\begin {proof} [Proof of Theorem \ref{thm:main-disjoint-unitdisks-preprocessing}]
We follow the proof of Theorem~\ref{thm:main-disjoint-unitdisks-existence}, which is constructive. First, we compute the $\RR_i$ sets for all $i$ using Lemma~\ref {lem:comprri}.
Once we have the sets $\RR_i$, we need to construct the $\Xi_i$, after which we simply concatenate the $\Xi_i$ in order.
Each $\Xi_i$ is constructed using Lemma~\ref{lem:line}, which in turn projects $\RR_i$ onto a line and then applies Lemma~\ref {lem:interval-windows}, whose run time is proportional to the output sequence length. As Theorem~\ref{thm:main-disjoint-unitdisks-existence} states that $|\Xi|=O(n)$, the overall running time is dominated by the $O(n\log n)$ time to construct all the $\RR_i$ sets. 
\end {proof}


\subsection{Preprocessing Overlapping Unit Disks}
\label {sec:2d-prep-overlapping}

We will now prove Theorem~\ref{thm:main-general-unspecific} for overlapping unit disks. 

\begin {theorem} \label{thm:main-unitdisks-existence-and-computation}
  Let $\RR$ be a set of $n$ unit disks in the plane of ply $\Delta$.
  There exists a $(3\sqrt2,12\Delta)$-\nice quarter-hull-supersequence $\Xi$ of $\RR$ of size $O(\Delta^2 n)$.
  $\Xi$ can be computed in $O(\Delta n(\Delta + \log n))$ time.
\end {theorem}

The main ideas and structure of the arguments remain the same; hence we only describe the main differences.
Most importantly, Lemma~\ref {lem:line} from Section~\ref {sec:2d-prep-disjoint} fails for overlapping disks, which implies that we need to construct our sequences differently.
We will set out to prove the following.

\begin{lemma} \label {lem:line-overlapping}
  Let $\ell$ be a line, and $\RR$ a set of unit disks of ply $\Delta$, which all have their centers within distance $d$ to the line  $\ell$.
  There is a sequence $\Xi$ of at most $(8\Delta
  \lceil d \rceil +8\Delta)|\RR|$ disks such that:
    for any point set $P \from \RR$,
    the order of the points in $P$ projected on $\ell$
    is a subsequence of $\Xi$.
Moreover, when considering the intervals resulting from the projection of $\RR$ onto $\ell$,
$\Xi$ is a $(3, (4\Delta\lceil d \rceil+4\Delta))$-\nice sorting-supersequence.
\end{lemma}

In order to prove the above lemma, we need several components.
First,
we will need the following lemma concerning sorting supersequences in 1d from \cite{1d-arxiv}:

\begin{lemma}[\cite{1d-arxiv}[Lemma 13]]\label{lem:interval-windows}
  Let $\cal I$ be a set of $n$ unit intervals in $\R$ of ply $\Delta$.
  There is a $(3, 2\Delta)$-\nice sorting-supersequence $\Xi$ of $\cal I$ with at most $4\Delta n$ intervals.
\end{lemma}


Second,
to prove Lemma~\ref{lem:line-overlapping},
we will also require a bound on the number of disks that can locally project onto a line.

\begin {lemma} \label {lem:project}
  Let $\ell$ be a line, and $\RR$ a set of unit disks of ply $\Delta$, which all have their centers within distance $d$ to the line  $\ell$.
  The ply of the projected intervals on $\ell$ is at most $2\Delta \lceil d \rceil + 2\Delta$.
\end {lemma}
\begin {proof}
  We need to find a bound on the ply of the projected disks in $\cal I$.
  The maximum ply is given by $\Delta$ times the number of disjoint disks that fit inside a rectangle of dimensions $4 \times (2 d + 2)$. Refer to Figure~\ref {fig:fitrectangle}. Specht~\etal~\cite {SPECHT201358} show that this number is at most $2 \lceil d \rceil + 2$.
\end {proof}

Now, we are ready to prove Lemma~\ref {lem:line-overlapping}.

\begin{proof} [Proof of Lemma~\ref {lem:line-overlapping}]
  W.l.o.g. assume $\ell$ is horizontal.
  We will project all disks to $\ell$ and consider the resulting set of unit intervals $\cal I$.
  By Lemma~\ref {lem:project}, we can bound the ply of the projected disks in $\cal I$ by $\Delta' = 2 \lceil d \rceil + 2$.  
  Then, we apply Lemma~\ref {lem:interval-windows} to obtain a sequence $\Xi$ of at most $(8\Delta\lceil d \rceil+8\Delta)|\RR|$ intervals, which translates directly to the required sequence of disks.
  Moreover, Lemma~\ref {lem:interval-windows} also tells us that $\Xi$ is 
  $(3, (4\Delta\lceil d \rceil+4\Delta))$-\nice.
\end{proof}

\begin {figure}[t] 
\centering
\includegraphics{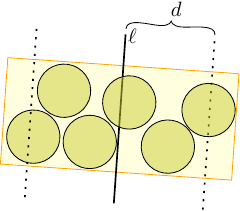}
\caption {The maximum ply of a set of projected disjoint unit disks is given by exactly the number of disjoint unit disks that fit inside a rectangle of dimensions $4 \times (2 d + 2)$.}
\label {fig:fitrectangle}
\end {figure}

Using Lemma~\ref{lem:line-overlapping} in place of Lemma~\ref {lem:line}, all our algorithms extend straightforwardly to overlapping disks of bounded ply. In particular, the definitions of strips, inner and outer disks, \nice sequences all remain the same.

However, we can no longer derive
{Corollary} \ref {cor:diskstrips}
and
{Corollary} \ref {cor:ply}
from
Lemma~\ref {lem:intersection_bound}.
Instead, we have the following more general bounds:

\begin {lemma} \label {lem:diskstrips}
  Consider the strip hull of a set $\RR$ of unit disks of ply $\Delta$.
  Every unit disk intersects at most $14\Delta-7$ strips.
\end {lemma}  

\begin {proof}
  Let $\cal B$ be the set of boundary disks of $\RR$,
  and let $\cal I$ be a maximum independent set (in the intersection graph) of $\cal B$.
  By {Corollary} \ref {cor:diskstrips},
  every unit disk intersects at most $7$ strips of the strip hull of $\cal I$. 
  Clearly every vertex of the strip hull of $\cal I$ is also a vertex of the strip hull of $\cal B$; furthermore, since every disk in $\cal B$ intersects at least one disk in $\cal I$, there can be at most $2\Delta-2$ additional vertices between any pair of adjacent vertices on the strip hull of $\cal B$; that is, each strip in of the strip hull of $\cal I$ corresponds to a sequence of at most $2\Delta-1$ strips in the strip hull of $\cal B$.
  The lemma follows.
\end {proof}

\begin {lemma} \label {lem:ply}
  Consider the strip hull of a set $\RR$ of unit disks of ply $\Delta$.
  The arrangement of strips has ply at most $8\Delta-4$.
\end {lemma}  
\begin {proof}
  Let $\cal B$ and $\cal I$ be as in the proof of Lemma~\ref {lem:diskstrips}.
  By {Corollary} \ref {cor:ply},
  Any point in the plane is contained in at most $4$ strips of the strip hull of $\cal I$, and thus in at most $8\Delta-4$ strips of the strip hull of $\cal B$.
\end {proof}

\begin {proof} [Proof of  Theorem~\ref{thm:main-unitdisks-existence-and-computation}]
The proof remains the same as the disjoint disk case, except that our bounds on $\sum |\RR_i|$ and $|\Xi_i|$ increase by $\Delta$, and the smoothness constants change. Specifically, 
using Lemma~\ref{lem:diskstrips} in place of 
Corollary~\ref{cor:diskstrips} in the proof of Theorem~\ref{thm:main-disjoint-unitdisks-existence}, we now obtain the bound
$\sum |\RR_i| \in O(\Delta n)$. 

For each strip $s_i$ separately, we define a sequence $\Xi_i$ consisting of possibly reoccurring disks of $\RR_i$, according to Lemma~\ref {lem:line-overlapping}. As the disks of $\RR_i$ have their centers within distance 2 of the line supporting the strip $s_i$, by Lemma~\ref {lem:line-overlapping} we have $|\Xi_i|\leq (8\cdot 2+8)\Delta|\RR_i| = O(\Delta|\RR_i|)$.
Since we concatenate the $\Xi_i$ to obtain $\Xi$, we have $ |\Xi|  \in O(\Delta^2 n)$ as required.

Furthermore, Lemma~\ref {lem:line-overlapping} states that when considering the intervals resulting from the projection of $\RR$ on $\ell$, $\Xi_i$ is a $(3, (4\cdot 2+4)\Delta) = (3,12\Delta)$-\nice
sorting-supersequence. 
Note that the angle between $\ell$ and a line of slope $-1$ is at most $45^\circ$; therefore, the distance $\alpha$ increases by at most a factor $\sqrt 2$. Moreover, $\beta$ will not increase when going from the projection on the line back to the plane.
We conclude that each $\Xi_i$ is $(3\sqrt2, 12\Delta)$-\nice.

Using Lemma~\ref {lem:ply} in place of Corollary \ref{cor:ply}
in the proof of Theorem~\ref{thm:main-disjoint-unitdisks-preprocessing},
we now have an arrangement of pseudodisks of ply $O(\Delta)$, which, using again the bound by Sharir~\cite{s-ksacs-91}, has complexity at most $O(\Delta n)$.
The rest of the proof remains the same.
\end{proof}


\subsection{Preprocessing Overlapping Disks of Varying Radii} 

Note that, unlike for overlapping disks, there is no reason a priori to expect that disks of arbitrary {\em  radii} cannot be handled.
However, the machinery developed is this paper relies heavily on the convexity of the spine hull, which we lose in this case (Figure~\ref{fig:strips-and-spines} (c)).

Nonetheless, here we show that our results do extend to the case when the disks have varying radii when there is a bounded ratio $\rho$ between the smallest and largest disk.

The basic idea is that we can simply enlarge all disks about their centers so that they have a uniform radius of $\rho$. 
Doing so will potentially increase the ply, and we now upper bound this value. So consider any point $x$ in the plane. Observe that any disk that overlaps $x$ after enlargement must have its center within a ball of radius $\rho$ about $x$.  As the enlargement does not move the disks centers, this implies any such disks must be entirely contained in a ball of radius $2\rho$ about $x$ before enlargement. A simple volumetric argument implies that the maximum number of unit balls of ply $\Delta$ that are entirely contained in a ball of radius $2\rho$  is at most $\Delta (\pi (2\rho)^2/ \pi(1)^2) = 4\Delta \rho^2$. (A more careful argument should improve the factor of 4.)

\begin {theorem} \label{thm:arbitrarydisks-preprocessing}
  Let $\RR$ be a set of $n$ disks in the plane, with minimum radius $1$ and maximum radius $\rho$, and let $\Delta$ be the ply of $\RR$.
  A $(O(1),O(\rho^2\Delta))$-\nice quarter-hull-supersequence $\Xi$ of $\RR$ 
  of size $O(\rho^4\Delta^2 n)$ can be computed 
  in $O(\rho^2\Delta n(\rho^2\Delta + \log n))$ time.
\end {theorem}

\section{The Reconstruction Phase} \label{sec:2d-rec}

In this section we present our reconstruction algorithm and then discuss how to extend it to achieve potential sublinear reconstruction by marking disks. Note the results in this section do not require the disks be unit disks nor that they have bounded ply. 
Indeed, we only know how to construct an $(\alpha,\beta)$-\nice supersequence for unit disks of bounded ply; however, we believe it is still of value to have a reconstruction algorithm for the more general case.
We will first consider {\em standard} reconstruction (that is, aiming for linear time) in Section~\ref {sec:2d-rec-standard}, and then discuss sublinear reconstruction in Section~\ref{sec:2d-sub}.

\subsection{Standard Reconstruction}
\label{sec:2d-rec-standard}

In this section, we consider the problem of recovering the true quarter hull from $\Xi$, given a realization $P \from \RR$. 
We again start by replacing the disks in $\Xi$ by their corresponding points; now we have a sequence $X$ of points that contains the points on the quarter hull in counterclockwise order as a subsequence. Whether we can recover the quarter hull of such a sequence in linear time in general is equivalent to Open Problem~\ref {prob:2d}.
Here, we additionally assume the sequence $X$ is $(\alpha, \beta)$-\nice for constants $\alpha$ and $\beta$.

\paragraph*{Algorithm for a Single Quadrant}

We start with the following observation.

\begin {observation}\label{obs:firstlast}
  We can assume that $X$ starts with the rightmost point and ends with the topmost point.
\end {observation}

\begin {proof}
  We can find the rightmost and topmost points in $X$ in linear time by scanning $X$ once.
  Note that there might be multiple copies of these points.
  We delete the prefix of $X$ before the first occurrence of the rightmost point, and we delete the suffix of $X$ after the last occurrence of the topmost point.
  The resulting sequence must still contain the quarter hull of $X$ as a subsequence; moreover, deleting points cannot violate that $X$ is $(\alpha,\beta)$-\nice.
\end {proof}

We now describe our algorithm for extracting the quarter hull.
It is an adaptation of the algorithm of Andrew~\cite {ANDREW1979216} (which is in turn an adaptation of Graham's scan), but we don't have a guarantee that the points come completely sorted. However, whenever we encounter a point that is not in the correct sorted order, we argue that we can determine whether it should still be inserted or can be safely discarded in $O(\alpha \cdot \beta)$ time be checking only the last $\alpha\cdot\beta$ edges of the hull so far.

\begin {lemma}\label{lem:quartersub}
Let $X$ be a $(\alpha, \beta)$-\nice sequence where the first point is the rightmost, the last is the topmost, and the quarter hull is a subsequence of $X$.
We can compute the quarter hull of $X$ in $O(\alpha\beta |X|)$ time.
\end {lemma}
\begin {proof}
We process the points in order, maintaining  a convex chain, $C$, of a subset the points processed so far, where initially $C$ is just the rightmost point. So let $q$ be the next point to be processed, and let $p$ be the last (i.e.\ topmost) point on $C$. If $q=p$, then we skip $q$ and move on to processing the next point. Otherwise, 
If the $y$-coordinate of $q$ is greater than or equal to that of $p$, then we apply a standard Graham scan insert (i.e. pop from $C$ until it is a left turn to $q$, and then insert $q$ at the end of the chain). Finally, if the $y$-coordinate of $q$ is less than that of $p$ then there are two cases. Let $\ell$ be the line of slope $-1$ through $p$, and let $proj(q)$ be the orthogonal projection of $q$ onto $\ell$. If $||proj(q)-p|| > \alpha$ then we pop $p$ from $C$, and we restart the process of trying to insert $q$. Otherwise, we check the last $\alpha \cdot \beta$ edges of $C$ in order from left to right, where if for any edge $q$ lies above the supporting line of that edge then we delete all vertices of $C$ starting from the left endpoint of that edge and going to the left, after which we then perform a Graham scan insert of $q$ into the remaining chain. If $q$ does not lie above any such edge, then we skip $q$ and move on to processing the next point.

For the running time, recall that in the last case when processing a point we might perform a constant time check on each of the last $\alpha \cdot \beta$ edges of $C$. Thus over all points this cost is $O(\alpha\beta |X|)$. This dominates the running time, as by the same analysis as in standard Graham scan, the remaining operations when processing a point take time proportional the number of deletions, and points are deleted at most once. 

As for correctness, first observe that in the case that $||proj(q)-p|| > \alpha$, then as $X$ is an $(\alpha,\beta)$-\nice sequence, by definition $p$ cannot be a vertex of the quarter hull. Thus the algorithm correctly deletes it, and so we will ignore this case in the remainder of the proof. 

Let  $C=c_1,\ldots c_m$ be the current convex chain at some point during the algorithm. Suppose that the prefix $c_1,\ldots, c_k$ is a prefix of the actual quarter hull. Then observe that this prefix can never be removed from $C$. Specifically, there are two cases when points are removed from $C$. The first case is by a Graham scan insertion, which by the correctness of Graham scan will not remove a vertex of the actual quarter hull. The second case is when if checking the last $\alpha\cdot \beta$ edges of $C$, the algorithm finds an edge where $q$ lies above its supporting line, however, such an edge cannot be found in the $c_1,\ldots, c_k$ portion of $C$ as it contains edges of the actual quarter hull.

Again let $C=c_1,\ldots,c_m$ be the convex chain at some iteration of the algorithm. Suppose that the $i-1$ prefix of $C$, namely $C_{i-1} = c_1,\ldots,c_{i-1}$, is the $i-1$ prefix of the quarter hull, and that the $i$th vertex of the quarter hull, call it $q$, is not in $c_i,\ldots, c_m$. In this case, we claim that if $q$ is processed next by the algorithm then afterwards $C_i$ with be the $i$th prefix of the quarter hull. 
As prefixes of the quarter hull are never removed as argued above, proving this claim inductively implies we correctly compute the quarter hull, as we know the quarter hull is a subsequence of $X$ and thus we must eventually process the next vertex of the quarter hull. (The base case holds as initially $C$ is the first point in $X$, which by the lemma statement is the first vertex of the quarter hull.)

To prove the claim we consider the different cases handled by the algorithm, in order. First, $q=c_m$ ($c_m$ is called $p$ above) is not possible since $q$ by assumption of the claim statement is not in $C$.  If the $y$-coordinate of $q$ is greater or equal to $c_m$ then we perform a Graham scan insert of $q$, which is guaranteed to place $q$ in $C$ after $c_{i-1}$ (since the fact that $q$ is above $c_m$ implies that $c_1,\ldots, c_{i-1}$ are the only vertices of the quarter hull in $C$). So suppose that the $y$-coordinate of $q$ is less than that of $c_m$.
Next, in the case that $||proj(q)-c_m|| > \alpha$, we already remarked that $c_m$ is correctly deleted. So suppose that $||proj(q)-c_m|| \leq \alpha$.

Now because $q$ is on the actual quarter hull, though it lies below $c_m$, it must lie to the right of $c_m$. In particular, it must lie above the chain $C$, i.e.\ in the region right of and below $c_m$, above and left of $c_1$, and above $C$. Thus $q$ lies above some edge of $C$, and since $||proj(q)-c_m|| \leq \alpha$, as $X$ is an $(\alpha,\beta)$-\nice sequence, the packing property implies it must be one of the last $\alpha \cdot \beta$ edges of $C$. Finally, after identifying and removing up to the right end of this edge, $q$ will make a right turn with the remainder of $C$ until $c_{i-1}$. Thus after applying a Graham scan update, $q$ will be placed after $c_{i-1}$ as claimed.
\end {proof}

Note in the above algorithm it is possible that $q$ is a duplicate of one of the vertices on the current chain (other than the last point $p$). This case is correctly handled, assuming in the later step we only check if $q$ is strictly above the supporting line of the edges, as this will discard $q$ which is the desired behavior. 

\paragraph*{Combining the Pieces}

\begin {proof} [Proof of Theorem~\ref{thm:mainunitdisks-reconstruction}]
If $\Xi$ is an $(\alpha,\beta)$-\nice quarter-hull-supersequence of $\RR$, then by definition for any $P\from\RR$, if we replace the regions from $\Xi$ by their corresponding points in $P$, then it results in an $(\alpha,\beta)$-\nice sequence such that the quarter hull of $P$ is a subsequence. 
Moreover, by Observation~\ref{obs:firstlast} we can assume the first point in the sequence is the rightmost and the last is the topmost.
Thus Lemma~\ref{lem:quartersub} directly states we can recover the quarter-hull of $P$ in $O(\alpha\beta|\Xi|)$ time.
\end {proof}

\subsection {Sublinear Reconstruction}
\label {sec:2d-sub}

If the reconstruction phase requires replacing each region in $\RR$ by the corresponding point in $P$, then this phase has an immediate linear time lower bound. However, if the location of the realization of a certain region does not affect the combinatorial structure of the output, then it may not be necessary to replace it with its point.
Following the approach of Van der Hoog~\etal~\cite{hkls-paip-19,hkls-pippf-22}, we thus modify the preprocessing framework by introducing another phase: the reconstruction phase is formally separated into a first subphase (that can take sublinear time), in which the auxiliary structure $\Xi$ is transformed into another structure $\Xi'$ which is combinatorially equivalent to the desired output $S(P)$, and a second subphase in which $\Xi'$ is actually transformed into $S(P)$ in linear time, if so desired.

An advantage of using supersequences as an auxiliary structures is that $\Xi'$ for our problem is simply a sequence of points and disks, which is the desired output convex hull when replacing the disks in the sequence with their corresponding realizations. As this second subphase is now trivial, we focus on the first subphase in the remainder of this section. 
This requires identifying the disks whose realizations do not affect the combinatorial structure of the convex hull, which are precisely the \emph{stable} disks defined in Section~\ref{sec:classifications}.

As argued in Section~\ref{sec:classifications}, there are two types of stable disks, namely stable guaranteed boundary disks and stable impossible interior disks.
First observe that stable impossible interior disks do not properly intersect any strip on the boundary (that is $I$ from Observations~\ref{obs:polygons} lies inside the chain of strips). Thus our constructed supersequence $\Xi$ already omits such disks, as it is constructed only from disks which intersect strips.

The remaining stable disks are all guaranteed boundary disks, which Lemma~\ref{lem:compute-classification} allows us to efficiently identify, and can be utilized as follows.

\begin {definition}
A {\em marked} convex-hull-supersequence $\Xi$ of $\RR$ is a sequence of (possibly recurring) disks of $\RR$, some of which may be marked, and such that, for any point set $P \from \RR$, the points on the convex hull of $P$ appear as a subsequence of $\Xi$, that contains all marked items of $\Xi$.
\end {definition}

A normal convex-hull-supersequence is a marked supersequence with no items marked. Our goal is to mark as many items as possible, since these may be skipped during reconstruction. By the discussion above, the largest set possible to mark is precisely the stable guaranteed boundary disks, which we will  mark using Lemma~\ref{lem:compute-classification}.

We now adapt the preprocessing algorithm to build a marked supersequence $\Xi$ as follows. Let the counterclockwise sorted order of boundary disks be $\mathcal{B}(\mathcal{R})=\{B_1,\ldots,B_k\}$. The marked disks in this ordering intuitively partition the disks into maximal length substrings of unmarked disks, which in turn partitions the strip hull into subchains. In particular, if $B_i$ is a marked disk then both of its adjacent strips ($B_{i-1}B_i$ and $B_{i}B_{i+1}$) are no longer considered, which is the desired behavior as recall the proof of Lemma~\ref{lem:compute-classification} argued no other disks intersected these strips, and because $B_i$ is stable the location of its realization does not affect the result on the substring ending at $B_{i-1}$ (and analogously for $B_{i+1}$).
Thus we can safely apply Theorem~\ref{thm:main-general-unspecific} independently to each resulting substring and concatenate the results in order, with the respective marked disks in between.

In order to be able to quickly skip over marked disks during the reconstruction phase, we can equip the resulting supersequence with pointers from each item to the next non-marked item. During the reconstruction phase we then also apply Lemma~\ref{lem:quartersub} separately to each subsequence of consecutive non-marked items, in time proportional to the total number of non-marked items in the list. 
This results in a mixed sequence of points and marked disks, which is guaranteed to be the convex hull of $P$.

\maarten {May need to be slightly precise here on how we combine splitting the problem by marked disks with splitting the problem on the leftmost, rightmost, topmost, or bottommost global points (which we might not know if they are in marked disks). I think it is clear enough that this will work out to be handwavy for now though - but for a full version we should probably work out the details.}




 \section {Conclusion \& Future Work} \label {sec:conclusion}

We introduced the idea of using supersequences as auxiliary structures in the preprocessing model for dealing with uncertain geometric data.
%
To follow up these first results using supersequences, we see three main lines of future work. 

First, our preprocessing results are currently restricted to bounded radius ratio disks of constant ply. In principle, it should be possible to generalize these results, e.g. to arbitrary disks, or fat regions, although most likely the {\em \nice} property in its current form would be too restrictive.

Second, our work exposed Open Problem~\ref {prob:2d}, which we believe is of independent interest, but also has direct implications on the preprocessing framework: a linear-time solution to this problem would immediately resolve the reconstruction problem on any supersequence, independent of the type of regions or any smoothness condition.

Finally, we only looked at constructing a supersequence for the convex hull. There are other natural problems that might fit in this framework, such as the Pareto Front, the traveling salesperson problem, shortest paths in polygons, and more. It would be interesting to study such problems in the same model.



\bibliographystyle{plainurl}
\bibliography{refs,geom,moregeom,maarten}

\appendix

\section{Prior Work on Uncertain Convex Hulls}\label{sec:rwork}

The convex hull is one of the oldest and most studied structures in computational geometry, dating back to Graham~\cite {g-eadch-72}; an overview of the literature on convex hull algorithms is beyond the scope of this work.
Within the relatively younger field of {\em uncertain} computational geometry, the convex hull is still one of the most-studied problems.


      Boissonnat and Lazard~\cite {bl-chbc-96} study the problem of finding the shortest convex hull of bounded curvature that contains a set of points, and they show that this is equivalent to finding the shortest convex hull of a set of imprecise points modelled as circles that have the specified curvature. They give a polynomial-time approximation algorithm.
      Goodrich and Snoeyink \cite {gs-spscp-90} also show how to minimize the area or perimeter of the polygon in $O(n^2)$ time.

      Van Kreveld and L\"offler~\cite {lk-alchip-08} study the problem of computing upper and lower bounds on the area or perimeter of the convex hull.
      Ju and Luo~\cite {jl-nacmp} improve one of these results, and also consider some variations in the model of imprecision.
      Some of the resulting problems here have also been studied in a different context.
      Mukhopadhyay \etal~\cite {mgr-isils-07,mkgb-ispls-08} study the largest area convex polygonal stabber of a set of parallel or isothetic line segments, and similarly, Hassanzadeh and Rappaport~\cite {hr-aafmpp-09,r-mptls-95} study the shortest perimeter convex polygonal stabber of a set of line segments.
      
      Many geometric algorithms return a certain region in the plane, or a subdivision of the plane into regions. In these cases, a natural way of representing imprecise output is by providing two boundaries of such a region: an inner boundary that encloses all points that are certain to be in the region, and an outer boundary that encloses all points that could be in the region. Such a pair of boundaries is sometimes referred to as a \emph {fuzzy} boundary~\cite {r-fguo-98}.
      Nagai and Tokura \cite {nt-teb-00} follow this approach for the convex hull by computing the union and intersection of all possible convex hulls.
      As imprecision regions they use disks and convex polygons, and they give an $O(n \log n)$ time algorithm for computing both boundaries.
      Sember and Evans~\cite {evans2011possible,Sember2011} define similar concepts.

Bruce~\etal~\cite {bruce2005efficient} consider update strategies for convex hulls for moving points, which is very similar to the static case with uncertainty.
They introduce (inspired in turn by Feder~\etal~\cite {10.1145/335305.335386}) a notion of uncertain points being {\em always}, {\em never}, or {\em dependent} on the convex hull, which is similar to part of our classification in Section~\ref {sec:classifications}.
However, their focus is on the number of {\em retrievals} rather than computation time: they assume an oracle is available which can be used to retrieve, for any region, the true location of its point, and aim to develop a strategy which, for any given instance, performs the smallest number of retrievals possible and still guarantees the correct output. 
Van der Hoog~\etal~\cite{hkls-pippf-22} later term such a strategy to be {\em instance optimal} if it performs no more retrievals than necessary by comparing to an adversarial strategy, though they focus on the Pareto front and not the convex hull.
Very recently, de Berg~\etal~\cite {deberg2025instanceoptimalimpreciseconvexhull} show how to do this in
 {\em instance optimal} time for the convex hull as well.

\end{document}